\documentclass{article}

\usepackage{geometry}                
\usepackage{graphicx}

\usepackage{epstopdf}

\usepackage{enumerate}

\usepackage{amsmath}
\usepackage{amssymb}
\usepackage{amsthm}
\usepackage{url}
\usepackage{color}

\usepackage{tikz-cd}

\usepackage{float}
\floatstyle{boxed} 
\restylefloat{figure}
\usepackage[section]{placeins}

\usepackage{chngcntr}
\counterwithin{figure}{section}

\usepackage{pifont}
\usepackage{stmaryrd}
\usepackage{dsfont}
\usepackage{textcomp}
\usepackage{verbatim}
\usepackage{accents}
\usepackage{wasysym}
\usepackage{csquotes}
\usepackage{booktabs}  

\newtheorem{thm}{Theorem}
\newtheorem*{thm*}{Theorem}
\newtheorem{prop}[thm]{Proposition}
\newtheorem{lemma}[thm]{Lemma}

\newtheorem*{mainthm*}{Main Theorem}
\newtheorem*{mainlemma*}{Main Lemma}
\newtheorem{cor}[thm]{Corollary}

\newtheorem*{claim*}{Claim}

\newtheorem{conj}[thm]{Conjecture}
\newtheorem*{conj*}{Conjecture}

\theoremstyle{definition}

\newtheorem{dfn}[thm]{Definition}
\newtheorem*{dfn*}{Definition}

\newtheorem*{fix*}{Fixation}

\newtheorem*{tgtfeat*}{Target features}

\newtheorem*{ass*}{Assumption}

\newtheorem{que}{Question}
\newtheorem*{que*}{Question}

\newtheorem*{chal*}{Challenge}

\newtheorem*{aim*}{Aim}

\newtheorem*{diag*}{Diagnosis}
\newtheorem*{mque*}{Main Question}

\newtheorem*{rem*}{Remark}

\newtheorem*{appr*}{Approach}
\newtheorem*{Copernican*}{Copernican Principle}
\newtheorem*{Correctness*}{Correctness Condition}

\newcommand{\Lang}{\mathcal{L}}

\newcommand{\df}{\mathrm{df}}

\newcommand{\dt}{\hspace{2pt}}

\newcommand{\PA}{\mathsf{PA}}

\usepackage{cancel}
\usepackage[linesnumbered,ruled,vlined]{algorithm2e}
\usepackage{bussproofs}

\newcommand{\ThA}[1]{\mathrm{Th}_{{#1}}}
\newcommand{\vdashNH}{\vdash_{\cancel{H}}}

\newcommand{\MT}{\mathrm{MT}}

\newcommand{\Appe}{\mathrm{App}}
\newcommand{\Repl}{\mathrm{Rep}}

\title{Intentic Semantics for Potentialist Truthmaking}
\author{Paul Gorbow}
\date{}                                           %

\begin{document}

\maketitle

\section{Complete Intentic Semantics for Predicate Logic}\label{Sec:Complex_intentic}

\subsection{Non-Hypothetical Logic}

Let $\Lang$ be a purely relational language. Recursively, let $\Lang_0 = \Lang$, and for each $i < \omega$, let $\Lang_{i+1}$ be $\Lang_i$ augmented with a constant symbol $c_{\varphi}$, for each $\varphi(x) \in \Lang_i$ whose sole free variable is $x$. Let $\Lang_\omega = \bigcup_{i < \omega} \Lang_i$. This work applies to both classical and intuitionistic first-order logic, but for expository concreteness we choose classical logic here. The deductive relation of classical logic is denoted $\vdash$. We employ the standard system of natural deduction (as defined e.g. in \cite[Ch. 2.1]{TS00}), and we treat LEM as an axiom schema. 
The logic obtained from the natural deduction by excluding ${\rightarrow} I$, and replacing $\vee E$ and $\exists E$ by $\exists E'$ and $\vee E'$ below, is called \emph{non-hypothetical logic}, since it refrains from reasoning on the basis of undischarged hypotheses. Its deductive relation is denoted $\vdashNH$.
\begin{align}
	&\frac{
		\varphi \vee \psi \quad \varphi \rightarrow \theta \quad \psi \rightarrow \theta
	}{
	\theta
	}
	\tag{$\vee E'$}
	\\[2ex]
	&\frac{
		\exists x \varphi(x) \quad \varphi(c_\varphi) \rightarrow \theta 
	}{
		\theta
	}
	\tag{$\exists E'$}
\end{align}

\subsection{Intentic States and their Truthmaking Semantics} \label{Subsec:Formal_TM}

For any finite set of formulas $d$, we define $[d]$ as the closure of $d$ under $\vdashNH$. Moreover, for $\varphi \in \Lang$, $d + \varphi$ is shorthand for $d \cup \{\varphi\}$.

Given $i < \omega$, a \emph{basic intentic state $b$ (over $\Lang$ of order $i$)} is a set of formulas of $\Lang_i$, such that $b = [d]$, for some finite subset $d \subseteq b$. Note that since $\vdashNH$ has the rule $\wedge E$, for any basic intentic state $b$, there is a single formula $\sigma_b$ (namely the conjunction of the formulas in $d$ above), such that $b = [\{\sigma_b\}]$. For the sake of notational convenience, we fix such a $\sigma_b$ for every basic intentic state $b$.

An \emph{intentic state $u$ (over $\Lang$)} is defined by $\in$-recursion as an ordered pair $\langle b(u), H(u) \rangle$, where $b(u)$ is a basic intentic state and $H(u)$ is a finite set of intentic states, such that for each $s \in H(u)$, there is $\varphi \in \Lang_\omega$, such that: 
\begin{align*}
	b(u) + \varphi = b(s)
\end{align*}
$b(u)$ is called \emph{the base state of $u$}. The elements of $H(u)$ are called \emph{hypothetical states} of $u$.

\begin{dfn}\label{Dfn:MT}
Let $u, v$ be intentic states over $\Lang$. We define the set $\MT(u)$, of formulas \emph{made true by $u$}, by $\in$-recursion as the $\subseteq$-least set, such that for any $\varphi, \theta \in \Lang_\omega$:
	\begin{enumerate}[{\rm (i)}]
		\item\label{Dfn:MT_base} $b(u) \subseteq \MT(u)$
		\item\label{Dfn:MT_closed_children} $\textrm{For any } s \in H(u) \textrm{, if } b(u) + \varphi \vdashNH b(s) \textrm{ and } \theta \in \MT(s) \textrm{, then } (\varphi \rightarrow \theta) \in \MT(u).$
		\item\label{Dfn:MT_closed_loc_log} $\MT(u) \textrm{ is closed under $\vdashNH$.}$
	\end{enumerate}
We write $u \Vdash \varphi$, and say that \emph{$u$ is a truth-maker of $\varphi$}, if $\varphi \in \MT(u)$.
\end{dfn}

Using the above notion of truth-making, we define two natural partial orders on intentic states, and a syntactic $\Vdash$-relation:

\begin{dfn}
	Let $u, v$ be intentic states over $\Lang$. 
	\begin{enumerate}[{\rm (a)}]
		\item We say that $v$ \emph{$\Vdash$-extends $u$}, denoted $u \leq_\Vdash v$, if $\MT(u) \subseteq \MT(v)$.
		\item By $\in$-recursion, we define that $v$ is a \emph{fine extension} of $u$, denoted $u \leq_F v$, if $b(u) = b(v)$, and for each $s \in H(u)$, there is $t \in H(v)$, such that $s \leq_F t$.	
		\item We write $\Gamma \Vdash \varphi$ if every intentic state $u$, such that $u \Vdash \Gamma$, has a fine extension $u'$, such that $u' \Vdash \varphi$.
	\end{enumerate}
\end{dfn}

\begin{lemma}\label{Lem:Fine_ext_makes_true}
	Let $u, v, w$ be intentic states over $\Lang$. If $u \leq_F v$ and $v \leq_F w$, then $u \leq_F w$. Moreover, if $u \leq_F v$, then $u \leq_\Vdash v$.
\end{lemma}
\begin{proof}
	The transitivity claim follows directly by $\in$-induction from the definition of $\leq_F$ and the transitivity the $=$-relation.
	
	Assume that $u \leq_F v$.
	To see that $u \leq_\Vdash v$, it suffices to verify that for any $\varphi, \theta \in \Lang_\omega$:
	\begin{enumerate}[{\rm (i)}]
		\item $b(u) \subseteq \MT(v)$
		\item $\textrm{For any } s \in H(u) \textrm{, if } b(u) + \varphi \vdashNH b(s) \textrm{ and } \theta \in \MT(s) \textrm{, then } (\varphi \rightarrow \theta) \in \MT(v).$
		\item $\MT(v) \textrm{ is closed under $\vdashNH$.}$
	\end{enumerate}
	The first holds by $b(v) = b(u)$ and (\ref{Dfn:MT_base}) in the definition of $\MT(v)$. The third holds by (\ref{Dfn:MT_closed_loc_log}) in the definition of $\MT(v)$. 
	
	To establish the second, we proceed by $\in$-induction on the structure of $u$. Since $u \leq_F v$, we have for each $s \in H(u)$ that there is $t_s \in H(v)$, such that $s \leq_F t_s$. The induction hypothesis is that $s \leq_\Vdash t_s$, for each $s \in H(s)$. Let $s \in H(u)$, and let $\varphi \in \Lang_\omega$, such that $b(u) + \varphi \vdashNH b(s)$. Suppose that $\theta \in \MT(s)$. By $b(u) = b(v)$ and $b(s) = b(t_s)$, we have $b(v) + \varphi \vdashNH b(t_s)$. Moreover, by $s \leq_\Vdash t_s$, we have $\theta \in \MT(t_s)$. So by (\ref{Dfn:MT_closed_children}) in the definition of $\MT(v)$, we have $(\varphi \rightarrow \theta) \in \MT(v)$, as desired. 
\end{proof}

Let $0 < n < \omega$. A \emph{path in $u$ (of length $n$)} is an $n$-tuple $\langle u_0, \cdots, u_n \rangle$, such that $u_0 = u$, and for each $0 \leq i < n$, we have $u_{i+1} \in H(u_i)$. 
$\langle u_0, \cdots, u_n \rangle$ is \emph{connected by} a sequence of formulas $\langle \varphi_1, \cdots, \varphi_n \rangle$ if for each $0 \leq i < n$, $b(u_i) + \varphi_{i+1} \vdashNH b(u_{i+1})$. 

\begin{dfn}
We define the following operations for any basic intentic state $c$, for any intentic states $u, u'$ over $\Lang$, any set $U$ of intentic states over $\Lang$, any $0 < n < \omega$, and any path $\langle u_0, \cdots, u_n \rangle$ in $u$: 
\begin{align*}
	\Repl(\langle u_0, \cdots, u_n \rangle, u') &=_\df
	\begin{cases}
		u' & \textrm{if $n=0$.} \\
		\big\langle b(u_0), (H(u_0) \setminus \{ u_1 \}) \cup \{ \Repl(\langle u_1, \cdots, u_n \rangle, u') \} \big\rangle & \textrm{if $n > 0$.} \\
	\end{cases} \\
	\Appe(u, U) &=_\df \big\langle b(u), H(u) \cup U \big\rangle \\
	\Appe(\langle u_0, \cdots, u_n \rangle, U) &=_\df \Repl \big( \langle u_0, \cdots, u_n \rangle, \Appe(u_n, U) \big) \\
	\mathrm{Boost}(u, c) &=_\df \big\langle b(u) \cup c, \{\mathrm{Boost}(v, c) \mid v \in H(u)\} \big\rangle
\end{align*}
\end{dfn}

\begin{prop}\label{Prop:State_Prop}
	Let $c$ be a basic intentic state over $\Lang$, let $u, u'$ be intentic states over $\Lang$, let $U$ be a set of intentic states over $\Lang$, let $0 < n < \omega$, and let $\langle u_0, \cdots, u_n \rangle$ be a path in $u$, connected by $\langle \varphi_1, \cdots, \varphi_n \rangle$.
	\begin{enumerate}[{\rm (a)}]
		\item\label{Prop:State_Prop_Rep} Let $v = \Repl(\langle u_0, \cdots, u_n \rangle, u')$. If $u_n \leq_F u'$, then $v$ is an intentic state, such that $u \leq_F v$.
		\item\label{Prop:State_Prop_App} Let $v = \Appe(\langle u_0, \cdots, u_n \rangle, U)$. If $\forall s \in U \dt \exists \varphi_s \in \Lang_\omega \dt \big( b(u_n) + \varphi_s = b(s) \big)$, then $v$ is an intentic state, such that $u \leq_F v$.
		\item\label{Prop:State_Prop_Boost} $\mathrm{Boost}(u, c)$ is an intentic state.
	\end{enumerate}
\end{prop}
\begin{proof}
	(\ref{Prop:State_Prop_Rep}) This is proved by induction on $n$. For $n=0$ it is immediate. The induction hypothesis is that it holds for $n-1$. Let $w = \Repl(\langle u_1, \cdots, u_n \rangle, u')$. By the induction hypothesis, $w$ is an intentic state, such that $u_1 \leq_F w$. So $b(u) + \varphi_1 \vdashNH b(w)$. Note that $v = \big\langle b(u), (H(u) \setminus \{u_1\}) \cup \{w\} \big\rangle$. Since $u_1 \leq_F w$, $v$ is an intentic state, such that $u \leq_F v$. 
	
	(\ref{Prop:State_Prop_App}) Let $s \in U$ be arbitrary. Since $s$ is an intentic state and $b(u_n) + \varphi_s = b(s)$, we have that $v$ is an intentic state. By (\ref{Prop:State_Prop_Rep}), it suffices to show that $u_n \leq_F \Appe(u_n, U)$. But this is immediate from the definitions of $\leq_F$ and $\Appe$.	
	
	(\ref{Prop:State_Prop_Boost}) Let $v \in H(u)$ be arbitrary. By induction, we may assume that $\mathrm{Boost}(v, c)$ is an intentic state. Moreover, since $b(v) = b(u) + \phi$, for some $\phi$, we have that $b\big(\mathrm{Boost}(u, c)\big) + \phi = b(u) + c + \phi = b(v) + c = b\big(\mathrm{Boost}(v, c)\big)$, as desired.
\end{proof}

\subsection{Soundness and Completeness}

\begin{thm}[Soundness]\label{Thm:Soundness}
	Let $\Gamma$ be a finite subset of $\Lang_\omega$ and $\varphi \in \Lang_\omega$: 
	\begin{align*}
		\Gamma \vdash \varphi \Rightarrow \Gamma \Vdash \varphi
	\end{align*}
\end{thm}
\begin{proof}
	Assume that $\Gamma \vdash \varphi$. Let $u$ be an intentic state, such that $u \Vdash \Gamma$. We need to construct a fine extension $v$ of $u$, such that $v \Vdash \varphi$. This is proved by induction on the complexity of natural induction proofs for $\vdash$. Thus, let $p$ be a proof witnessing $\Gamma \vdash \varphi$. Let $P$ be the set of proofs of the assumptions of the last rule, say $R$, applied in $p$. Each $q \in P$ witnesses a deductive relation $\Gamma_q \vdash \varphi_q$. By induction, we may assume that $\Gamma_q \Vdash \varphi_q$, for each $q \in P$. We proceed by cases, as to which rule $R$ is.
	
	Suppose that $R$ is ${\rightarrow}I$, with the conclusion $\varphi$ being $(\psi \rightarrow \theta)$. Then the sole proof $q \in P$ witnesses $\Gamma, \psi \vdash \theta$. By the induction hypothesis, $\langle b(u) + \psi, \varnothing \rangle$ has a fine extension $v_1$, such that $v_1 \Vdash \theta$. It now follows from Proposition \ref{Prop:State_Prop}(\ref{Prop:State_Prop_App}) that $v = \Appe(u, \{v_1\})$ is a fine extension of $u$, and it follows from (\ref{Dfn:MT_closed_children}) in the definition of $\MT(v)$ that $v \Vdash (\psi \rightarrow \theta)$, as desired.
	
	Suppose that $R$ is $\vee E$. We then have formulas $\psi, \theta$, a proof $q_1 \in P$ witnessing $\Gamma \vdash \psi \vee \theta$, a proof $q_2 \in P$ witnessing $\Gamma, \psi \vdash \varphi$, and a proof $q_3 \in P$ witnessing $\Gamma, \theta \vdash \varphi$. By the induction hypothesis, $u$ has a fine extension $v_1$, such that $v_1 \Vdash \psi \vee \theta$; $\big\langle b(u) + \psi, \varnothing \big\rangle$ has a fine extension $v_2$, such that $v_2 \Vdash \varphi$; and $\big\langle b(u) + \theta, \varnothing \big\rangle$ has a fine extension $v_3$, such that $v_3 \Vdash \varphi$. Now we let $v = \Appe\big( v_1, \{v_2, v_3\} \big)$. By Proposition \ref{Prop:State_Prop}(\ref{Prop:State_Prop_App}), we have $v_1 \leq_F v$. So by Lemma \ref{Lem:Fine_ext_makes_true}, $u \leq_F v$, and $v \Vdash \psi \vee \theta$. Moreover, by (\ref{Dfn:MT_closed_children}) in the definition of $\MT(v)$, we have $v \Vdash \psi \rightarrow \varphi$, and $v \Vdash \theta \rightarrow \varphi$. So by (\ref{Dfn:MT_closed_loc_log}) in the definition of $\MT(v)$, and the rule $\vee E$ in $\vdashNH$, we have $v \Vdash \varphi$.
	
	Suppose that $R$ is $\exists E$. We then have a formula $\psi(x)$, a proof $q_1 \in P$ witnessing $\Gamma \vdash \exists x \psi(x)$, and a proof $q'_2 \in P$ witnessing $\Gamma, \psi(a) \vdash \varphi$, where $a$ is a fresh variable not appearing freely in $\varphi$. Let $q_2$ be the proof obtained from $q'_2$ by replacing each occurrence of $a$ by a constant $c_\psi$ not appearing in $p$ nor in $b(u)$. Note that $q_2$ witnesses $\Gamma, \psi(c_\varphi) \vdash \varphi$. By the induction hypothesis, $u$ has a fine extension $v_1$, such that $v_1 \Vdash \exists x \psi(x)$; and $\big\langle b(u) + \psi(c_\psi), \varnothing \big\rangle$ has a fine extension $v_2$, such that $v_2 \Vdash \varphi$. Now we let $v = \Appe\big( v_1, \{v_2\} \big)$. By Proposition \ref{Prop:State_Prop}(\ref{Prop:State_Prop_App}), we have $v_1 \leq_F v$, and by Lemma \ref{Lem:Fine_ext_makes_true}, $u \leq_F v$. Moreover, by (\ref{Dfn:MT_closed_children}) in the definition of $\MT(v)$, we have $v \Vdash \exists x \varphi(x)$, and $v \Vdash \psi(c_\psi) \rightarrow \varphi$. So by (\ref{Dfn:MT_closed_loc_log}) in the definition of $\MT(v)$, and the rule $\exists E$ in $\vdashNH$, we have $v \Vdash \varphi$.
	
	Suppose that $R$ is any other rule of natural deduction. Then there are no undischarged assumptions to handle. So for each $q \in P$, there is $\psi_q$, such that $q$ witnesses $\Gamma \vdash \psi_q$. Thus, by the induction hypothesis, for each $q \in P$, there is a fine extension $v_q$ of $u$, such that $v_q \Vdash \psi_q$. Let $v = \Appe\big(u, \bigcup_{r \in P} H(v_r) \big)$. Note that for each $q \in P$, we have $v = \Appe\big(v_q, H(u) \cup \bigcup_{r \in P} H(v_r) \big)$. So by Proposition \ref{Prop:State_Prop}(\ref{Prop:State_Prop_App}), $v$ is a fine extension of $u$, and of $v_q$, for each $q \in P$. Thus, by Lemma \ref{Lem:Fine_ext_makes_true}, $v \Vdash \psi_q$, for each $q \in P$. Now, since $R$ is admitted by $\vdashNH$, we have by (\ref{Dfn:MT_closed_loc_log}) in the definition of $\MT(v)$ that $v \Vdash \varphi$.
\end{proof}

\begin{lemma}\label{Lem:vdash_MT}
	Let $u$ be an intentic state over $\Lang$. Then $b(u) \vdash \MT(u)$.
\end{lemma}
\begin{proof}
	Let $\varphi \in \MT(u)$. We need to show that $b(u) \vdash \varphi$. We proceed by induction on the complexity of $u$. The induction hypothesis is that for each $s \in H(u)$, each $\theta$ such that $b(s) = b(u) + \theta$, and each $\psi$ such that $s \Vdash \psi$, there is a proof $p_{s, \theta, \psi}$ witnessing $b(u), \theta \vdash \psi$.
	
	Since $\MT(u)$ is the least set closed under the conditions (\ref{Dfn:MT_base})--(\ref{Dfn:MT_closed_loc_log}) in its definition, and $\varphi \in \MT(u)$, we have that there is a proof $p_{u, \Delta, \varphi}$ witnessing $\Delta \vdashNH \varphi$, for some finite set of formulas $\Delta$, such that for each $\chi \in \Delta$, either of the following holds:
	\begin{itemize}
		\item $\chi \in b(u)$, or 
		\item $\chi$ is the formula $(\theta \rightarrow \chi')$, for some $\theta, \chi'$, such that there is $s \in H(u)$, such that $b(u) + \theta = b(s)$ and $s \Vdash \chi'$. 
	\end{itemize}
	
	In the former case, let $p_{u, \chi}$ be the trivial proof witnessing $b(u) \vdashNH \chi$. In the latter case, the proof $p_{s, \theta, \chi'}$ witnessing $b(u), \theta \vdash \chi'$ can be followed by an application of ${\rightarrow} I$ to yield a proof $p_{u, \chi}$ witnessing $b(u) \vdash \chi$. Thus, each assumption $\chi \in \Delta$ in the proof $p_{u, \Delta, \varphi}$ can be replaced by the proof $p_{u, \chi}$, to yield a proof $p_{u, \varphi}$ witnessing $b(u) \vdash \varphi$, as desired.
\end{proof}

\begin{cor}
	Let $u$ be an intentic state over $\Lang$. If $b(u)$ is $\vdash$-consistent, then $\MT(u)$ is $\vdash$-consistent.
\end{cor}

\begin{thm}[Completeness]\label{Thm:Completeness}
	Let $\Gamma$ be a finite subset of $\Lang_\omega$ and $\varphi \in \Lang_\omega$: 
	\begin{align*}
		\Gamma \vdash \varphi \Leftarrow \Gamma \Vdash \varphi
	\end{align*}
\end{thm}
\begin{proof}
	Assume that $\Gamma \Vdash \varphi$. Then there is a fine extension $u$ of the intentic state $\big\langle \ThA{\vdashNH}(\Gamma), \varnothing \big\rangle$, such that $u \Vdash \varphi$. Since $b(u) = \ThA{\vdashNH}(\Gamma)$, we have by Lemma \ref{Lem:vdash_MT} that $\ThA{\vdashNH}(\Gamma) \vdash \varphi$. So $\Gamma \vdash \varphi$, as desired. 
\end{proof}

\section{Complexity of non-hypothetical logic}

This section defends the following conjectures, by drafting an algorithm for deciding the former:

\begin{conj}\label{Conj:Decidable}
	Non-hypothetical logic is decidable over the axioms of $\PA$,\footnote{For technical convenience, we take $\PA$ to be Peano Arithmetic in a purely relational language, which is axiomatized in Figure \ref{Fig:Rel_PA} below.} and more ambitiously, over any fixed theorem of $\PA$ taken as an additional axiom.
\end{conj}

A recurring feature of the proof-search procedure drafted below is the restriction to finitely many constants. In the purely relational axiomatization of $\PA$ adopted here, there are no function symbols and hence no nontrivial closed terms; in particular, the language itself contains no constants beyond those explicitly introduced. As a result, proof search over the axioms of relational $\PA$ does not give rise to an unbounded supply of new terms. This finiteness plays a crucial role in ensuring that the search space remains controlled. The more ambitious extension of the conjecture to arbitrary fixed theorems of $\PA$ motivates allowing finitely many additional constants, because from the perspective of intentic semantics, these constants refer to hypothesized objects under consideration. The algorithms below are therefore formulated so as to make explicit where finiteness of constants is relied upon, anticipating this stronger setting.

Proof search in the non-hypothetical setting typically halts when encountering an implication. Outside of implication introduction $(\rightarrow I)$, it is difficult to see how a derivation of an implication could contribute substantively to a closed proof, since implication elimination presupposes the availability of an antecedent that cannot itself be established hypothetically. As a consequence, proof search does not branch beyond implications in a way that would generate genuinely new information. This has the further effect that attempts to reason by induction are effectively cut short: the inductive step presupposes the derivation of an implication of the form $\varphi(x) \rightarrow \varphi(S(x))$, which cannot be discharged non-hypothetically and therefore cannot drive an unbounded proof search. Although induction axioms are present, their effective use in proof search is therefore blocked by the non-hypothetical restriction.

Similarly, the application of conjunction elimination does not appear to be proof-theoretically productive in this setting. Since elimination rules cannot introduce information not already present in axioms or previously derived closed formulas, pursuing such steps does not expand the space of available conclusions and can therefore be safely disregarded in a terminating search procedure.

Taken together, these observations suggest that normal non-hypothetical proofs should exhibit a strong form of subformula discipline, with proof search confined to formulas already occurring (positively) in the axioms or in the goal formula. Establishing this rigorously requires structural proof-theoretic methods, in particular normalization results adapted to the non-hypothetical setting, which are expected to underpin the decidability conjectures stated above.

The following algorithms are intended to indicate the sources of finiteness underlying the conjecture rather than to constitute a fully specified or implementable decision procedure. They should be read as a schematic proof-search discipline rather than as a complete algorithm.


\begin{algorithm}[H]
	\SetAlgoLined
	\SetKwFunction{Uni}{Uni}
	\SetKwFunction{Axiom}{Axiom}
	\SetKwFunction{PosSubAxiom}{PosSubAxiom}
	\SetKwFunction{Prove}{Prove}
	\SetKwFunction{ProveE}{ProveElim}
	\SetKwFunction{ProveI}{ProveIntro}
	\SetKwFunction{Axiomatic}{Axiomatic}
	\SetKwProg{Fn}{Function}{:}{}
	\Fn{\Axiom{$\varphi$}}{\tcc{Decides whether $\varphi$ is an axiom, up to variable substitution and closure under universal quantification.}}
	\Fn{\PosSubAxiom{$\varphi$}}{\tcc{Decides whether $\varphi$ occurs syntactically positively as a subformula of an axiom, up to variable substitution.}}
	\Fn{\Axiomatic{$p$}}{\tcc{Takes a proof as input and decides whether all its assumptions are axioms.}
		\caption{Helpers assumed to be available}
	}
\end{algorithm}

\begin{algorithm}
	\SetAlgoLined
	\SetKwFunction{Uni}{Uni}
	\SetKwFunction{Axiom}{Axiom}
	\SetKwFunction{PosSubAxiom}{PosSubAxiom}
	\SetKwFunction{Prove}{Prove}
	\SetKwFunction{ProveE}{ProveElim}
	\SetKwFunction{ProveI}{ProveIntro}
	\SetKwFunction{Axiomatic}{Axiomatic}
	\SetKwProg{Fn}{Function}{:}{}
	\Fn{\Prove{$\varphi$}}{
		\tcc{Returns a normal natural deduction of $\varphi$ in $\vdashNH$ from a fixed decidable set of axioms, if any, else returns $\bot$.}
		\If{\Axiom{$\varphi$}}{
			\Return{$\varphi$}
		}
		\ElseIf{\Axiomatic{\ProveI{$\varphi$}}}{
			\Return{\ProveI{$\varphi$}}
		}
		\ElseIf{\Axiomatic{\ProveE{$\varphi$}}}{
			\Return{\ProveE{$\varphi$}}
		}
		\Else{
			\Return{$\bot$}
		}
		\caption{Recursive proof-search for formula $\varphi$}
	}
\end{algorithm}

\begin{algorithm}
	\SetAlgoLined
	\SetKwFunction{Uni}{Uni}
	\SetKwFunction{Axiom}{Axiom}
	\SetKwFunction{PosSubAxiom}{PosSubAxiom}
	\SetKwFunction{Prove}{Prove}
	\SetKwFunction{ProveE}{ProveElim}
	\SetKwFunction{ProveI}{ProveIntro}
	\SetKwFunction{Axiomatic}{Axiomatic}
	\SetKwProg{Fn}{Function}{:}{}
	\Fn{\ProveI{$\varphi$}}{
		\If{$\varphi \equiv \psi \wedge \theta$}{
			\Return{
				\begin{prooftree}
					\AxiomC{\Prove{$\psi$}}
					\AxiomC{\Prove{$\theta$}}
					\BinaryInfC{$\psi \wedge \theta$}
				\end{prooftree}
			}
		}
		\ElseIf{$\varphi \equiv \psi \vee \theta$}{
			\If{\Prove{$\psi$} $\neq \bot$}{
				\Return{
					\begin{prooftree}
						\AxiomC{\Prove{$\psi$}}
						\UnaryInfC{$\psi \vee \theta$}
					\end{prooftree}
				}
			}
			\ElseIf{\Prove{$\theta$} $\neq \bot$}{
				\Return{
					\begin{prooftree}
						\AxiomC{\Prove{$\theta$}}
						\UnaryInfC{$\psi \vee \theta$}
					\end{prooftree}
				}
			}
			\Else{\Return{$\bot$}}
		}
		\ElseIf{$\varphi \equiv \forall x \psi(x)$}{
			\Return{
				\begin{prooftree}
					\AxiomC{\Prove{$\psi(x)$}}
					\UnaryInfC{$\forall x \psi(x)$}
				\end{prooftree}			
			}
		}
		\ElseIf{$\varphi \equiv \exists x \psi(x)$}{
			\tcc{Proof search is restricted to the finitely many constants occurring in axioms or previously introduced during the search.}
			\For{each such constant $c_\psi$}{
				\If{\Prove{$\psi(c_\psi)$} $\neq \bot$}{
					\Return{
						\begin{prooftree}
							\AxiomC{\Prove{$\psi(c_\psi)$}}
							\UnaryInfC{$\exists x \psi(x)$}
						\end{prooftree}
					}
				}
			}
		}
		\caption{Proof-search step by introduction rules}
	}
\end{algorithm}

\begin{algorithm}
	\SetAlgoLined
	\SetKwFunction{Uni}{Uni}
	\SetKwFunction{Axiom}{Axiom}
	\SetKwFunction{PosSubAxiom}{PosSubAxiom}
	\SetKwFunction{Prove}{Prove}
	\SetKwFunction{ProveE}{ProveElim}
	\SetKwFunction{ProveI}{ProveIntro}
	\SetKwFunction{Axiomatic}{Axiomatic}
	\SetKwProg{Fn}{Function}{:}{}
	\Fn{\ProveE{$\varphi$}}{
		\If{{\bf not} \PosSubAxiom{$\varphi$}}{
			\Return{$\bot$}
		}
		\ElseIf{\Axiom{\Uni{$\varphi$}}}{
			\Return{
				\begin{prooftree}
					\AxiomC{\Uni{$\varphi$}}
					\UnaryInfC{$\vdots$}
					\UnaryInfC{$\varphi$}
				\end{prooftree}
			}
		}
		\Else{
			\For{``each of the finitely many elimination possibilities given the finitely many axioms and the finitely many constants''}{``the proof-search for that elimination possibility''}
		}
		\caption{Proof-search step by elimination rules}
	}
\end{algorithm}

\begin{figure}[h!]\label{Fig:Rel_PA}
	\centering
	
	\vspace{1em}
	
	\textbf{Axioms of Classicality:}
	\begin{align*}
		&\varphi \vee \neg\varphi \tag{LEM}
	\end{align*}
	
	\textbf{Axioms of Equality:}
	\begin{align*}
		&\forall x\, (x = x) \tag{E1} \\
		&\forall x\, \forall y\, (x = y \rightarrow y = x) \tag{E2} \\
		&\forall x\, \forall y\, \forall z\, (x = y \land y = z \rightarrow x = z) \tag{E3} \\
		&\forall x\, \forall x'\, \forall y\, \forall y'\, \forall z\, \forall z'\, (x = x' \land y = y' \land z = z' \land \mathsf{Add}(x, y, z) \rightarrow \mathsf{Add}(x', y', z')) \tag{E4} \\
		&\forall x\, \forall x'\, \forall y\, \forall y'\, \forall z\, \forall z'\, (x = x' \land y = y' \land z = z' \land \mathsf{Mul}(x, y, z) \rightarrow \mathsf{Mul}(x', y', z')) \tag{E5} \\
		&\forall x\, \forall x'\, \forall y\, \forall y'\, (x = x' \land y = y' \land S(x, y) \rightarrow S(x', y')) \tag{E6}
	\end{align*}
	
	\vspace{1em}
	\textbf{Functionality and Totality Axioms:}
	\begin{align*}
		&\forall x\, \exists y\, S(x, y) \tag{SF1} \\
		&\forall x\, \forall y\, \forall z\, (S(x, y) \land S(x, z) \rightarrow y = z) \tag{SF2} \\
		&\forall x\, \forall y\, \exists z\, \mathsf{Add}(x, y, z) \tag{AF1} \\
		&\forall x\, \forall y\, \forall z\, \forall z'\, (\mathsf{Add}(x, y, z) \land \mathsf{Add}(x, y, z') \rightarrow z = z') \tag{AF2} \\
		&\forall x\, \forall y\, \exists z\, \mathsf{Mul}(x, y, z) \tag{MF1} \\
		&\forall x\, \forall y\, \forall z\, \forall z'\, (\mathsf{Mul}(x, y, z) \land \mathsf{Mul}(x, y, z') \rightarrow z = z') \tag{MF2}
	\end{align*}
	
	\vspace{1em}
	\textbf{Proper Axioms of Relational PA:}
	\begin{align*}
		&\forall x\, \neg S(x, 0) \tag{S1} \\
		&\forall x\, \forall y\, \forall z\, (S(x, z) \land S(y, z) \rightarrow x = y) \tag{S2} \\
		&\forall x\, \mathsf{Add}(x, 0, x) \tag{A1} \\
		&\forall x\, \forall y\, \forall y'\, \forall z\, \forall z'\, (S(y, y') \land \mathsf{Add}(x, y, z) \land S(z, z') \rightarrow \mathsf{Add}(x, y', z')) \tag{A2} \\
		&\forall x\, \mathsf{Mul}(x, 0, 0) \tag{M1} \\
		&\forall x\, \forall y\, \forall y'\, \forall z\, \forall u\, (S(y, y') \land \mathsf{Mul}(x, y, z) \land \mathsf{Add}(z, x, u) \rightarrow \mathsf{Mul}(x, y', u)) \tag{M2} \\
		&\varphi(0) \land \forall y\, \forall y'\, (S(y, y') \wedge \varphi(y) \rightarrow \varphi(y')) \rightarrow \forall x\, \varphi(x) \tag{IND}
	\end{align*}
	
	\caption{Axioms of Peano Arithmetic in a relational signature}
\end{figure}

\section{Comparison and Discussion}
\subsection{Relation to Linnebo's Potentialism}

\begin{que}
	How does the structure of all intentic states over $\Lang$, endowed with the relations $\leq_\Vdash$ and $\leq_F$, relate to the bi-modal potentialism of Linnebo? 
\end{que}


We show that this structure on intentic states satisfies the same axioms that Linnebo's bi-modal potentialism require of the two accessibility relations of its possible worlds frames:

\begin{thm}
	Let $\mathcal{S}$ be the structure of all intentic states over $\Lang$, endowed with the relations $\leq_\Vdash$ and $\leq_F$. We have:
	\begin{enumerate}[{\rm (a)}]
		\item $\leq_\Vdash$ is a preorder.
		\item $\leq_F$ is a {\em locally directed} preorder, that is, it is a preorder such that 
		\[\forall u, v, w \in \mathcal{S} \,\big( (u \leq_F v \wedge u \leq_F w) \rightarrow \exists s \in \mathcal{S} \, (v \leq_F s \wedge w \leq_F s) \big).\]
		\item The preorders relate as follows:
		\[ \forall u, v, w \in \mathcal{S} \,\big( (u \leq_\Vdash v \wedge u \leq_F w) \rightarrow \exists s \in \mathcal{S} \, (v \leq_F s \wedge w \leq_\Vdash s) \big). \]
	\end{enumerate}
\end{thm}
\begin{proof}
	(a) is immediate from the definition of $\leq_\Vdash$. 
	
	(b): It is immediate from Lemma \ref{Lem:Fine_ext_makes_true} that $\leq_\Vdash$ is a preorder. Suppose that $u, v, w \in \mathcal{S}$ and $u \leq_F v \wedge u \leq_F w$. Let $s = \langle b(u), H(v) \cup H(w) \rangle$. It is immediate from the definition of $\leq_F$ that $v \leq_F s \wedge w \leq_F s$.
	
	(c): Suppose that $u, v, w \in \mathcal{S}$ and $u \leq_\Vdash v \wedge u \leq_F w$. Let $s = \mathrm{App}\big(v, \{\mathrm{Boost}(w, b(v))\}\big)$. By Proposition \ref{Prop:State_Prop}(\ref{Prop:State_Prop_Boost}), $s$ is a fine extension of $v$. Now, suppose that $\phi \in \MT(w)$. Then $\phi \in \MT\big(\mathrm{Boost}(w, b(v))\big)$. Hence, $\sigma_{b(w)} \rightarrow \phi \in \MT(s)$. But $b(w) = b(u)$, so since $u \leq_\Vdash v$, we have $\sigma_{b(w)} \in \MT(v) \subseteq \MT(s)$. Therefore, $\phi \in \MT(s)$, as desired.
\end{proof}




\end{document}